\ifx\STACS\undefined%
\documentclass[12pt]{article}%
\else%
\documentclass[runningheads,a4paper]{llncs}
 \fi

\ifx\STACS\undefined
\newcommand{\InConfVer}[1]{}
\newcommand{\InFullVer}[1]{#1}
\else
\newcommand{\InConfVer}[1]{#1}
\newcommand{\InFullVer}[1]{}
\fi

   \usepackage[breaklinks]{hyperref} 
\usepackage{graphicx}
\usepackage[active]{srcltx} 
\InFullVer{%
   \usepackage{fullpage}%
}%
\InFullVer{%
   \usepackage{euscript}%
}%

\usepackage{xspace}%
\usepackage{caption}%
\usepackage{color}%
\usepackage{amsmath}%
\usepackage{paralist}%
\usepackage{picins}%
\usepackage{boxedminipage}
\usepackage{enumerate}
\usepackage{picins}
\usepackage{multirow}
\InConfVer{%
   \usepackage{microtype}
}
\usepackage[noend]{algorithmic}
\InFullVer{%
   \usepackage{ntheorem}%
   \theoremseparator{.}%
}
\InConfVer{%
   \def\ProofInAppendix{TRUE}
}

\usepackage{proofapnd}

\InFullVer{%
\setlength{\textwidth}{6.5in}
\setlength{\topmargin}{0.0in}
\setlength{\headheight}{0in}
\setlength{\headsep}{0.0in}
\setlength{\textheight}{9in}
\setlength{\oddsidemargin}{0in}
\setlength{\evensidemargin}{0in}
}
\newcommand{\WSPD}{\Term{WSPD}\index{well-separated pair
      decomposition}%
   \xspace}

\newcommand{\PRPOTSAT}{\Term{PRPOT3SAT}\xspace}
\newcommand{\ThreeSAT}{\Term{3SAT}\xspace}
\newcommand{\True}{\Term{TRUE}\xspace}
\newcommand{\False}{\Term{FALSE}\xspace}

\newcommand{\LowerBoundedCenter}   {\PStyle{{Lower{}Bounded{}Center}}\xspace}
\newcommand{\kCenter} {\PStyle{k-Center}\xspace}
\newcommand{\lbc}{\PStyle{LBC}\xspace}

\newcommand{\PTAS}{\Term{P{T}AS}\xspace}

\renewcommand{\th}{\si{th}\xspace}

\providecommand{\si}[1]{#1}

\newcommand{\AlinaThanksInner}[1]{%
   {Department of Computer Science; %
      University of Illinois; %
      201 N. Goodwin Avenue; %
      Urbana, IL, 61801, USA; %
      {\tt \si{ene}1\atgen{}\si{uiuc}.\si{edu}}; {\tt
         \si{\url{http://www.cs.uiuc.edu/\string~ene1/}}.} #1}}
\newcommand{\AlinaThanks}[1]{%
   \thanks{\AlinaThanksInner{#1}}}

\newcommand{\BenThanksInner}[1]{{Department of Computer Science;
      University of Illinois; %
      201 N. Goodwin Avenue; %
      Urbana, IL, 61801, USA; %
      {\tt \si{raichel}2\atgen{}\si{uiuc}.\si{edu}}; %
      {\tt \url{\si{http://www.cs.uiuc.edu/\string~\si{raichel2}}}}%
      . %
      #1}}

\newcommand{\BenThanks}[1]{\thanks{\BenThanksInner{#1}}}

\newcommand{\atgen}{\symbol{'100}}
\newcommand{\SarielThanksInner}[1]{{Department of Computer Science;
      University of Illinois; 201 N. Goodwin Avenue; Urbana, IL,
      61801, USA; {\tt \si{sariel}\atgen{}\si{uiuc.edu}}; %
      {\tt \url{http://www.uiuc.edu/\string~sariel/}.} #1}}
\newcommand{\SarielThanks}[1]{\thanks{\SarielThanksInner{#1}}}

\definecolor{blue25}{rgb}{0,0,0.55}%
\newcommand{\emphic}[2]{%
   \textcolor{blue25}{%
      \textbf{\emph{#1}}}%
   \index{#2}}

\newcommand{\emphi}[1]{\emphic{#1}{#1}}
\newcommand{\pth}[2][\!]{#1\left({#2}\right)}

\definecolor{red25}{rgb}{0.4,0,0.0}%

\newcommand{\PStyle}[1]{\textcolor{red25}{\textrm{\textsf{#1}}}}

\InFullVer{%
   \newtheorem{theorem}{Theorem}[section]
   \newtheorem{lemma}[theorem]{Lemma}%
   \newtheorem{definition}[theorem]{Definition}
   
   \newtheorem{corollary}[theorem]{Corollary}

   {\theorembodyfont{\rm} }
   {\theorembodyfont{\rm} }
}
\InFullVer{
   \newtheorem{problem}[theorem]{Problem}%
}

\newcommand{\distChar}{\mathsf{d}}%
\newcommand{\distX}[2]{\distChar\pth{#1, #2}}

\algsetup{indent=2em}

\newcommand{\distE}[1]{\left\| {#1}  \right\|}

\newcommand{\Term}[1]{\textsf{#1}}%

\newcommand{\APXHard}{\Term{AP{X}-hard}\xspace}%
\newcommand{\NPHard}{\Term{NP-Hard}\xspace}%
\renewcommand{\P}{\Term{P}\xspace}%
\newcommand{\NP}{\Term{NP}\xspace}%

\renewcommand{\Re}{{\rm I\!\hspace{-0.025em} R}}

\newcommand{\MakeSBig}{\rule[0.0cm]{0.0cm}{0.38cm}} 

\newcommand{\seclab}[1]{\label{sec:#1}}
\newcommand{\secref}[1]{Section~\ref{sec:#1}}

\newcommand{\thmlab}[1]{{\label{theo:#1}}}%
\newcommand{\thmref}[1]{Theorem~\ref{theo:#1}}%

\providecommand{\deflab}[1]{\label{def:#1}}
\newcommand{\defref}[1]{Definition~\ref{def:#1}}

\newcommand{\corlab}[1]{\label{cor:#1}}
\newcommand{\corref}[1]{Corollary~\ref{cor:#1}}

\newcommand{\apndlab}[1]{\label{apnd:#1}}
\newcommand{\apndref}[1]{Appendix~\ref{apnd:#1}}

\newcommand{\lemlab}[1]{\label{lemma:#1}}
\newcommand{\lemref}[1]{Lemma~\ref{lemma:#1}}

\newcommand{\figlab}[1]{\label{figure:#1}}
\newcommand{\figref}[1]{Figure~\ref{figure:#1}}

\InFullVer{%
   \newenvironment{proof}{\trivlist\item[]\emph{Proof}:}%
   {\unskip\nobreak\hskip 1em plus 1fil\nobreak%
      \myqedsymbol
      \parfillskip=0pt%
      \endtrivlist}

   {\unskip\nobreak\hskip 1em plus 1fil\nobreak%
      \myqedsymbol
      \parfillskip=0pt%
      \endtrivlist}%
}

\newcommand{\myqedsymbol}{\rule{2mm}{2mm}}
\newcommand{\cardin}[1]{\left| {#1} \right|}

\newcommand{\permut}[1]{\left\langle {#1} \right\rangle}%

\newcommand{\eps}{\varepsilon}%

\newcommand{\remove}[1]{}

\newcommand{\floor}[1]{\left\lfloor {#1} \right\rfloor}

\newcommand{\etal}{\textit{et~al.}\xspace}

\usepackage{pifont}%

\usepackage[symbol*]{footmisc}

\DefineFNsymbols*{stars}[text]{%
   {\ding{172}} %
   {\ding{173}} %
   {\ding{174}} %
   {\ding{175}} %
   {\ding{176}} %
   {\ding{177}} %
   {\ding{178}} %
   {\ding{179}} %
   {\ding{180}} %
   }
\setfnsymbol{stars}

\newcommand{\pnt}{{\mathsf{p}}}%
\newcommand{\pntA}{{\mathsf{q}}}%

\newcommand{\PntSet}{\mathsf{P}}%

\renewcommand{\th}{th\xspace}

\newcommand{\centers}{C}%
\newcommand{\lb}{\lambda}

\newcommand{\dset}{\mathcal{D}}
\newcommand{\cset}{\mathcal{C}}
\newcommand{\gset}{\mathcal{G}}
\newcommand{\ropt}{r_{opt}}
\newcommand{\roptX}[2]{\ropt\pth{#1, #2}}

\newcommand{\net}{\Algorithm{Net}\xspace}

\DefineNamedColor{named}{RedViolet} {cmyk}{0.07,0.90,0,0.34}
\newcommand{\AlgorithmI}[1]{{\textcolor[named]{RedViolet}{\texttt{\bf{#1}}}}}
\newcommand{\Algorithm}[1]{{\AlgorithmI{#1}\index{algorithm!#1@{\AlgorithmI{#1}}}}}

\newcommand{\polylog}{\mathrm{polylog}}

\newcommand{\netX}[2]{\net\pth{#2, #1}}

\newenvironment{proofext}{\noindent{\em Proof of
   }}{\hfill{\hfill\myqedsymbol}}

\newcommand{\ProofFake}[1]{ 
   {\emph{Proof:}} %
   #1
   \hfill{\myqedsymbol}
}
\newcommand{\Grid}{\mathsf{G}}

\begin{document}

\title{Fast Clustering with Lower Bounds:\\ 
No Customer too Far, No Shop too Small}

\InConfVer{%
   \titlerunning{Fast Clustering with Lower Bounds}%
   \authorrunning{A. Ene, S. Har-Peled, and B. Raichel}
   \institute{}
}%

   \author{%
      Alina Ene%
      \AlinaThanks{Work on this paper was partially supported by NSF
         grants CCF-0728782 and CCF-1016684.}%
      \and%
      Sariel Har-Peled%
      \SarielThanks{Work on this paper was partially supported by a
         NSF AF awards CCF-0915984 and CCF-1217462.}%
      \and%
      Benjamin Raichel%
      \BenThanks{Work on this paper was partially supported by a NSF
         AF award CCF-0915984.}%
   }%


\maketitle

\begin{abstract}
	We study the \LowerBoundedCenter (\lbc) problem, which is a
	clustering problem that can be viewed as a variant of the
	\kCenter problem. In the \lbc problem, we are given a set of
	points $P$ in a metric space and a lower bound $\lambda$, and the
	goal is to select a set $C \subseteq P$ of centers and an
	assignment that maps each point in $P$ to a center of $C$ such
	that each center of $C$ is assigned at least $\lambda$ points.
	The price of an assignment is the maximum distance between a
	point and the center it is assigned to, and the goal is to find a
	set of centers and an assignment of minimum price.
	We give a constant factor approximation algorithm for the \lbc
	problem that runs in $O( n \log n)$ time when the input
	points lie in the $d$-dimensional Euclidean space $\Re^d$, where
	$d$ is a constant.  We also prove that this problem cannot be
	approximated within a factor of $1.8 - \epsilon$ unless $\P = \NP$
	even if the input points are points in the Euclidean plane
	$\Re^2$.
  \InConfVer{
  }
\end{abstract}


\section{Introduction}

Clustering is a practical and well studied problem in computer
science.  Work on clustering varies greatly based on one's choice of
how to measure the quality of the clustering.  Clustering is a variant
of unsupervised learning and has been widely studied \cite{dhs-pc-01}. In
clustering, the input is a set of points $\PntSet$ in a metric space
and we are interested in partitioning it into ``nice'' clusters. What
a nice cluster is depends on the application at hand, and the resources
available.

Clustering measures that have algorithms with guaranteed performance
include:
\begin{compactitem}
    \item \textbf{$k$-median clustering}: Here one has to choose $k$
    center points, and the price of the clustering is the sum of
    distances of the points to their respective closest center. There
    is considerable work on this problem both in the general and
    geometric settings \cite{arr-asekm-98,cgts-caamp-99,agkmp-lshkm-01,c-kmchd-06,c-cfaak-08}.

    \item \textbf{$k$-means clustering}: Here, the price is the sum of
    squared distances of a point to its nearest center. This
    clustering is very popular in practice as it has a very simple
    heuristic that works reasonably well. There is also some work on
    understanding the performance of this heuristic 
    \cite{l-lsqp-82,hs-hfkmm-05,amr-sakmm-11}.  
    It can also be approximated using local search
    \cite{kmnpsw-lsaak-04}. In lower dimensional Euclidean settings it
    can be approximated using coresets \cite{hk-sckmk-07}. In higher
    dimensions it can be approximated in linear time if the number of
    clusters and quality of approximation, $\eps$, is fixed \cite{kss-ltasc-10}
    (the constant in the running time depends exponentially on $k$ and
    $\eps$).
   
    \item \textbf{$k$-center clustering}: Here, the price is the
    maximum distance a point has to travel to its closest center.
    Gonzalez \cite{g-cmmid-85} showed a $2$ approximation algorithm
    (for general metric spaces) with running time $O(nk)$. This was
    later improved by Feder and Greene \cite{fg-oafac-88}, who showed
    an $O\pth{n \log k}$ time algorithm, which is optimal in the
    comparison model.  A linear time algorithm is known if is allowed
    to use randomization and the floor function \cite{h-cm-04}.  Feder
    and Greene also showed that it can not be approximated within
    $1.8$, in the plane, unless $\P = \NP$.
\end{compactitem}

\paragraph*{Lower bounded center clustering.}
Here we are interested in a variant of the \kCenter clustering
problem where one is allowed to open as many centers as one likes, as
long as each center gets assigned enough clients.  Note that in this
variant, called \LowerBoundedCenter (\lbc), one has to choose not
only where to place the centers, but also how to assign the clients
to the open centers.

The \lbc problem is quite natural in the following sense.  Suppose
you are trying to decide where and how many stores, wireless towers,
post offices, etc., to open.  In order to open a store in a given
area there needs to be a sufficiently large client base.  In order to
keep the clients happy, you wish to minimize the maximum
client-center distance.  There is now no longer a limit, $k$, on the
number of stores you can open.  Instead you open as many stores as
you can to make your clients happy, subject to the constraint that
each store you open is profitable.

Aggarwal \etal \cite{apftkkz-aac-10} consider \LowerBoundedCenter in
general metric spaces, and provide a matching upper and lower
approximation bound of 2.  The algorithm they provide for the upper
bound as described takes at least $\Omega(n^3)$ time, and involves
several maximum flow computations.

\paragraph*{Nets and clustering.} %
An \emphi{$r$-net}\footnote{One should not confuse this concept with
   the $\eps$-net concept used in the context of VC-dimension theory.}
is a subset $C \subseteq \PntSet$, such that:
\begin{inparaenum}[(i)]
    \item the distance between any pair of points of $C$ is larger than
    $r$, and
    \item the distance of any point of $\PntSet$ to its nearest
    neighbor in $C$ is at most $r$.
\end{inparaenum}
An $r$-net is a good representation of the point set if we care about
the resolution $r$.

Interestingly, but somewhat tangentially, one can compute a greedy
permutation of the points of $\PntSet = \permut{\pnt_1, \ldots,
   \pnt_n}$, and radii $r_1 \geq r_2 \geq \cdots \geq r_n$, such that
$\PntSet_i = \permut{\pnt_1, \ldots, \pnt_i}$ is (approximately) an
$r_i$-net of $\PntSet$, for all $i$. As such, the greedy permutation
is a good way to encode nets of a point-set in all resolutions.  This
permutation can be computed in $O(n \log n)$ time for finite metric
spaces having low doubling dimension \cite{hm-fcnld-06} (as such, this
algorithm also works in low dimensional Euclidean space). See
\cite{h-gaa-11} for more information.

\subsection{Our Results}

We start by observing, in \secref{compNets}, that $r$-nets can be
computed in linear time in low dimensional Euclidean space.  This
observation is implicit in previous work \cite{h-cm-04}, but it is
worth bringing it to the forefront, as it serves as our basic
building block (in particular, one can not just use the algorithm of
Har-Peled \cite{h-cm-04} as it has restrictions on the number of
clusters it can handle). 
In \secref{lbc} it is shown that using this and well-separated pairs
decompositions (\WSPD{}s) leads to an $O(n \log n)$ time
$(4+\eps)$-approximation algorithm, see \lemref{wspd:lbc}.  
This near linear running time is a significant speed up from the work of 
Aggarwal \etal \cite{apftkkz-aac-10}, though at the cost of 
loosing another factor of 2 in the approximation.

Underscoring the importance of the linear time net computation 
presented here, this net computation has since been used 
to get a randomized expected linear time algorithm which 
achieves constant factor approximations to a more general class of 
problems \cite{hr-nplta-13}.  

As such, one of the main contributions of the current paper is in 
providing a hardness reduction proving that no \PTAS exists for \lbc unless 
$\P = \NP$, even in the plane (inspired to some extent by
the hardness proof in \cite{fg-oafac-88}). Note that the lower bound
of 2 provided in \cite{apftkkz-aac-10} for general metric spaces does
not apply to the specific case of $\Re^2$.

\paragraph*{Paper organization.} %
We start in \secref{compNets} by showing how to compute nets
efficiently.  In \secref{lbc} we present the clustering algorithm
using net computation.  The hardness of approximation is proved in
\secref{hardness}.



\section{Computing Nets Quickly for a Point Set %
   in $\Re^d$}
\seclab{compNets}

Here we show how to compute nets quickly. Note that our basic
approach is implicit in previous work \cite{h-cm-04}. We emphasize
that we cannot use the clustering algorithm of Har-Peled
\cite{h-cm-04} directly for our purposes, since the algorithm does
not quite compute what we need and the algorithm's running time is
linear only if the number of clusters is small; the algorithm that we
give in this paper runs in linear time for any number of clusters.

\begin{definition}%
    \deflab{net}%
    For a point set $\PntSet$ in a metric space with a metric
    $\distChar$, and a parameter $r>0$, an $r$-\emphi{net} of $\PntSet$ is a
    subset $\cset \subseteq \PntSet$, such that
    \begin{inparaenum}[(i)]
        \item for every $\pnt,\pntA \in \cset$, $\pnt \neq \pntA$, we
        have that $\distX{\pnt}{\pntA} > r$, and
        \item for all $\pnt\in \PntSet$, we have that
        $\distX{\pnt}{\cset}=\min_{\pntA\in \cset} \distX{\pnt}{\pntA}
        \leq r$.
    \end{inparaenum}
\end{definition}

There is a simple algorithm for computing $r$-nets.  Namely, let all
the points in $\PntSet$ be initially unmarked.  While there remains an
unmarked point, $p$, add $p$ to the set of centers $\cset$, and mark
it and all other points within distance $\leq r$ from $p$ (i.e. we are
scooping away balls of radius $r$).  By using grids and hashing one
can modify this algorithm to run in linear time.

\begin{lemma}%
    \lemlab{net}%
    Given a point set $\PntSet\subseteq \Re^d$ of size $n$ and a
    parameter $r>0$, one can compute an $r$-net for $\PntSet$ in
    $O(n)$ time.
\end{lemma}

\vspace{-0.5cm}
\parpic[r]{%
   \begin{minipage}{0.25\linewidth}%
       %
       %
       \vspace{-.03cm}
       \includegraphics[width=.95\linewidth]{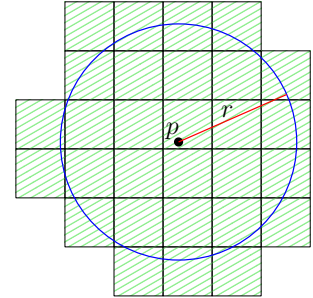}\\%
       \vspace{-0.85cm}
           \captionof{figure}{\small\newline Neighborhood of $p$.}%
           \figlab{planar}%
   \end{minipage}
}%


\newcommand{\GridNbrChar}{\mathsf{N}}
\newcommand{\GridNbr}[2]{\mathsf{N}_{\leq #2}\pth{#1}}
\newcommand{\GridNbrCell}[1]{\mathsf{N}\pth{#1}}
\newcommand{\gidX}[1]{\mathrm{id}\pth{#1}} 
\ProofFake{%
    Let $\Grid$ denote the grid in $\Re^d$ with side length $\Delta =
    r/\pth{2\sqrt{d}}$.  First compute for every point $\pnt \in
    \PntSet$ the grid cell in $\Grid$ that contains $\pnt$; that is,
    the cell containing $\pnt$ is uniquely identified by the tuple of
    integers $\gidX{\pnt} = \pth[]{\MakeSBig\! \floor{\pnt_1/\Delta},
       \ldots, \floor{\pnt_d/\Delta}}$, where $\pnt = \pth{\pnt_1,
       \ldots, \pnt_d} \in \Re^d$.  Let $\gset$ denote the set of
    non-empty grid cells of $\Grid$.  Similarly, for every non-empty
    cell $g \in \gset$ we compute the set of points of $\PntSet$ which
    it contains.  This task can be performed in linear time using
    hashing and bucketing assuming the floor function can be performed
    in constant time, as using hashing we can store a grid cell
    $\gidX{\cdot}$ in a hash table and in constant time hash each
    point into its appropriate bin.  For a point $\pnt \in \PntSet$
    let $\GridNbr{\pnt}{r}$ denote the set of grid cells in distance
    $\leq r$ from $\pnt$, which is the \emphi{neighborhood} of $\pnt$.
    Observe that $\cardin{\GridNbr{\pnt}{r}} =
    O\pth{(2r/(r/2\sqrt{d})+1)^d} = O\pth{(4\sqrt{d}+1)^d } = O(1)$.
   
    Scan the points of $\PntSet$ one at a time, and let $\pnt$ be the
    current point.  If $\pnt$ is marked then move on to the next
    point.  Otherwise, add $\pnt$ to the set of net points, $\cset$,
    and mark it and each point $\pntA \in \PntSet$ such that
    $\distX{\pnt}{\pntA} \leq r$.  Since the cells of $\GridNbr{\pnt}{r}$
    contain all such points, we only need to check the lists of points
    stored in these grid cells.  At the end of this procedure every
    point is marked.  Since a point can only be marked if it is in
    distance $\leq r$ from some net point, and a net point is only
    created if it is unmarked when visited, this implies that $\cset$
    is an $r$-net.
   
    For the running time, observe that a grid cell, $c$, has its list
    scanned only if $c$ is in the neighborhood of some created net
    point.  From the discussion above we know that there are $O(1)$
    cells which could contain a net point $\pnt$ such that $c\in
    \GridNbr{\pnt}{r}$.  Also, we create at most one net point per
    cell since the diameter of a grid cell is strictly smaller than
    $r$.  Therefore $c$ had its list scanned $O(1)$ times.  Since the
    only real work done is in scanning the cell lists and since the
    cell lists are disjoint, this implies an $O(n)$ running time
    overall.
}%

\bigskip

Observe, that the closest net point, for a point $\pnt \in \PntSet$,
must be in one of its neighborhood grid cells. Since every grid cell
can contain only a single net point, it follows that in constant time
per point of $\PntSet$, one can compute its nearest net point.  We
thus have the following.

\begin{corollary}%
    \corlab{valid}%
    In $O(n)$ time one can not only compute an $r$-net, but also
    compute for each center the set of points of $\PntSet$ for which 
    it is the nearest center.
\end{corollary}


\section{Approximation Algorithm for \lbc}
\seclab{lbc}


Our aim is to get an efficient, small constant factor approximation for
the following problem.

\begin{problem}[\LowerBoundedCenter/\lbc] Let $\PntSet$ be a set of
    $n$ points in $\Re^d$, and let $\lb>0$ be an integer parameter.
    We wish to find a set of \emphi{centers} $\centers\subseteq
    \PntSet$, and an assignment of the points in $\PntSet$ to the
    centers in $\centers$, such that every center in $\centers$ gets
    at least $\lb$ points of $\PntSet$ assigned to it.

    The \emphi{price} of the clustering is the maximum distance of a
    point of $\PntSet$ to its nearest center in $\centers$. The
    optimal \lbc clustering is the one minimizing this maximum price.
\end{problem}

We now present two approximation algorithms for \lbc.  
The first is an $O(n^{4/3}\polylog(n))$ time $4$-approximation algorithm 
and the second an $O(n \log( n/\eps ))$ time $(4+\eps)$-approximation algorithm.  
One can get a $2$-approximation using flows, as shown in 
\cite{apftkkz-aac-10}, though the running time is significantly slower.


\subsection{Approximation Algorithms}
Let $P$ be a set of points and let $C \subseteq P$ be a set of
centers. The \emph{nearest center} assignment for $P$ and $C$ is the
assignment that maps each point of $P$ to its closest center in $C$.

\begin{definition}
    Given a point set $\PntSet$, let $\netX{\PntSet}{r}$ denote an
    $r$-net $N$ of $\PntSet$, along with the nearest center assignment
    of the points of $\PntSet$ to the centers of $N$.
\end{definition}

Note that, using the algorithm of \corref{valid}, one can compute
$\netX{\PntSet}{r}$ for a point set in $\Re^d$ in linear time.

\medskip%
\begin{definition}%
    \deflab{valid}%
	Let $\PntSet$ be a set of points. The net $\netX{\PntSet}{r}$ is
	\emphi{valid}, with respect to a lower bound $\lb$, if every
	center in $\netX{\PntSet}{r}$ is assigned at least $\lb$ points
	by the nearest center assignment for $\PntSet$ and
	$\netX{\PntSet}{r}$.
\end{definition}

Note that if $\netX{\PntSet}{r}$ is valid then we have a solution to
\LowerBoundedCenter of price $\leq r$.

\begin{lemma}%
    \lemlab{four}%
	Let $\pth{\PntSet,\lb}$ be an instance of $\lbc$.  Let $\cset$ be
	an $\alpha\ropt$-net for $\PntSet$, for $\alpha\geq 4$, where
	$\ropt = \roptX{\PntSet}{\lb}$ is the price of the optimal
	solution for the given \lbc instance.  Then the centers along
	with the nearest center assignment is an $\alpha$-approximate
	solution to \lbc.

    In particular, if $\PntSet$ is a set of $n$ points in $\Re^d$, and
    given a distance $x$, then one can decide, in linear time, if
    $\ropt \leq x$ or $\ropt > x/4$ (if $\ropt \in [x/4,x]$ either
    answer might be returned).
\end{lemma} 

\parpic[r]{%
   \begin{minipage}{0.25\linewidth}%
       \includegraphics[width=.98\linewidth]{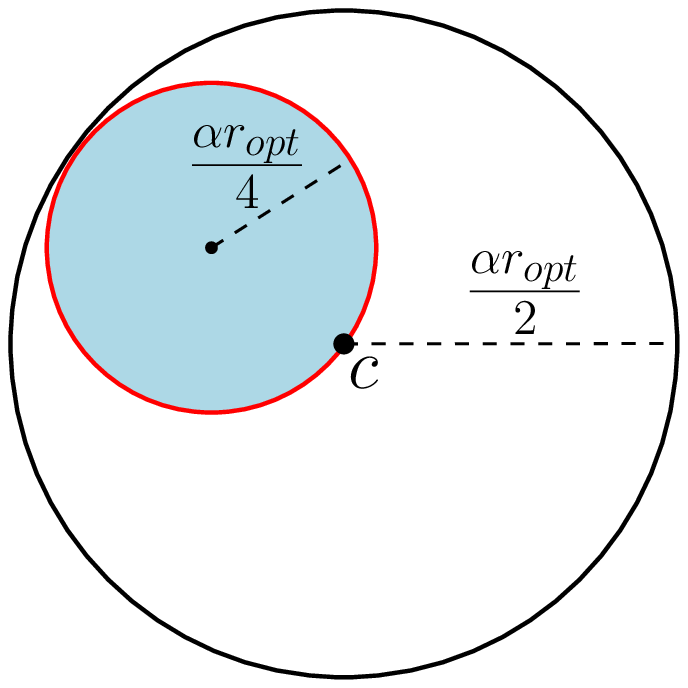}
   \end{minipage}}

\ProofFake{%
   We need to verify that the lower bound is satisfied for each
   center.  By the definition of a net, the distance between any two
   centers in $\cset$ is greater than $\alpha\ropt$.  Therefore, if we
   create a ball of radius $\alpha\ropt/2$ around each center, then
   these balls will be pairwise disjoint and thus all points in these
   balls are assigned to the center of the ball that contains them
   (recall that the assignment is a closest center assignment).
   Consider such a ball around a center $c$.  If this ball contains
   an entire ball from the optimal solution, then $c$'s lower bound
   will be satisfied.  So consider the ball which covers $c$ in the
   optimal solution.  This ball has radius $\ropt\leq \alpha\ropt/4$,
   since $\alpha\geq 4$.  However any ball covering $c$ of radius
   $\leq \alpha\ropt/4$ is entirely contained in $c$'s ball of radius
   $\alpha\ropt/2$.

   The second part follows by using the algorithm of \corref{valid} to
   decide whether or not $\netX{\PntSet}{x}$ is valid, see
   \defref{valid}.
}%
\bigskip%

The lemma above implies a general line of attack.  First find $\ropt$
(or an upper bound to $\ropt$ that is a constant factor larger), and
then compute a $4\ropt$-net.

Consider an instance of \lbc and let $\dset$ be a set of positive
real numbers.  Let $\alpha, \beta$ be two values in $\dset$ such that
$\alpha<\beta$.  We say that $[\alpha, \beta]$ is an \emphi{atomic}
interval if $(\alpha, \beta) \cap \dset = \emptyset$.  Furthermore,
we say $[\alpha, \beta]$ is an \emphi{active} interval if
$\net(4\alpha, \PntSet)$ is invalid and $\net(\beta, \PntSet)$ is
valid. By \lemref{four}, if $[\alpha,\beta]$ is active then
$\alpha<\ropt\leq \beta$.  

\subsubsection{Using Exact Distances}

Let $\dset$ be the set of all $O(n^2)$ distances between pairs of
points.  Note that $\ropt \in\dset$. We can now do a binary search for
$\ropt$, using median selection at each iteration to find a
candidate value $x$ (in the current set of still relevant values), and
then use \lemref{four} to decide if the value $4x$ is too large or too
small.  Median selection and pivoting on a set of size $s$ takes
$O(s)$ time.  Since the size of the set decreases by a half each time,
we get a geometric series and so the total work involved for median
selection and pivoting is $O(\cardin{\dset})$.  The recursion takes
$O(\log n)$ rounds to bottom out and $O(n)$ time (by \corref{valid})
is spent in each round computing and verifying a net.  Therefore the
overall running time is $O(n^2+n\log n) = O(n^2)$.

The running time can be further improved using distance selection.
One can compute the $k$\th smallest distance in $\dset$ in $O(n\log n
+ n^{2/3}k^{1/3}\log^c n)$ time, for some small constant $c$
(Chan \cite{c-esd-01} shows a randomized algorithm for $c=5/3$).  In
particular, the median can be computed in $O(n^{4/3} \polylog(n))$
time and so performing a binary search using median selection leads to
the following result.

\begin{lemma}
    Given an instance $\pth[]{\PntSet,\lb}$ of \LowerBoundedCenter in
    $\Re^d$, one can compute a $4$-approximation in $O(n^{4/3}
    \allowbreak \polylog(n))$ time.
\end{lemma}

\subsubsection{Using Approximate Distances}

\begin{lemma}%
    \lemlab{interval}%
    Given an instance $\pth[]{\PntSet,\lb}$ of \LowerBoundedCenter in
    $\Re^d$, and an interval $[x,y]$ that contains $\ropt =
    \ropt\pth{\PntSet, \lb}$, and $\eps>0$, then one can compute a 
    $(4+\eps)$-approximation to $\ropt$ and its associated clustering
    in $O\pth{ n \log \pth{ 2 \eps^{-1} y/x }}$ time.
\end{lemma}

\begin{proof}
	Let $r_i = x(1+\eps/8)^i$, for $i=0, \ldots, M =
	\floor{\log_{1+\eps/8} \pth{16y /x}} = O\pth{ \eps^{-1} y/x }$.
	Using binary search and the algorithm of \lemref{four}, find an
	index $i$, such that $\netX{\PntSet}{r_{i+1}}$ is valid but
	$\netX{\PntSet}{r_{i}}$ is invalid (see \defref{valid}); if no
	$r_i$ is invalid then $x=\ropt$. Clearly, $r_{i+1}$ is the
	required approximation. Since this procedure uses the algorithm
	of \lemref{four} $O\pth{ \log M}$ times, the claim follows.
\end{proof}

The above does not solve the general problem, as the spread of the
interval containing the solution (i.e., $y/x$) might be arbitrarily
large.  However, using \WSPD{}s \cite{ck-dmpsa-95} the running time
can be improved to $O(n\log n+n\log (1/\eps))$.  In particular, one
can use \WSPD{}s to compute, in $O(n\log n)$ time, a set $\dset$ of
distances, of size $O(n)$, such that $\dset$ contains an active atomic
interval $[x, y]$, and $y\leq cx$ for some constant $c$ (for more
details, see \cite{dhw-afdrc-12}).  Given $\dset$, such an interval
can be found with binary search in $O(n\log n)$ time using
\lemref{four} as the decision procedure.  Now, using \lemref{interval}
on this interval results in the desired approximation.

\begin{lemma}%
    \lemlab{wspd:lbc}%
    Given an instance of \LowerBoundedCenter in $\Re^d$, one can
    compute a $(4+\eps)$-approximation in $O(n \log( n/\eps ) )$ time.
\end{lemma}


\section{Hardness of Approximation of \lbc}
\seclab{hardness}

We now prove that it is not possible to approximate
\LowerBoundedCenter within a factor of $\sqrt{13}/2$ for points in
the plane, unless $\P=\NP$.  This will be done via a reduction from
Positive Rectilinear Planar 1 in 3 \ThreeSAT (\PRPOTSAT), which is
known to be $\NPHard$ \cite{mr-mwtnph-08}.  Our reduction is similar
in spirit to the reduction of Feder and Greene \cite{fg-oafac-88}. In
particular, we prove the following.

\begin{theorem}%
    \thmlab{hardness}%
	There is no polynomial time algorithm which approximates
	\LowerBoundedCenter on points in the Euclidean plane $\Re^2$
	within a factor of $\sqrt{13}/2 \approx 1.80$ unless $\P=\NP$.
\end{theorem}
\subsection{The Setting}
\parpic[r]{\includegraphics[width=.4\linewidth]{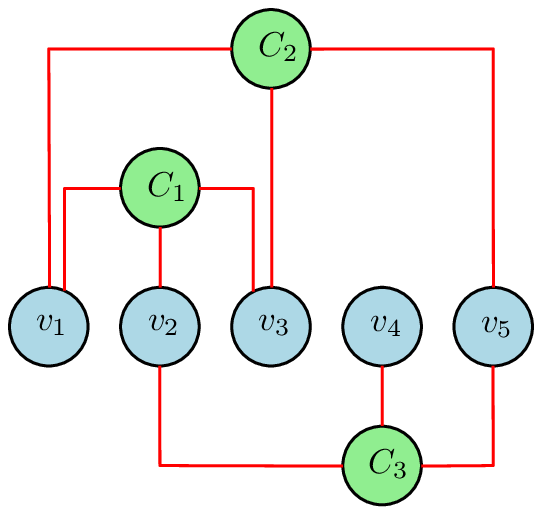}}
Let $\phi$ be a boolean formula in $3CNF$ form (i.e. a conjunction of
clauses, where each clause is a disjunction of three literals).  We say
$\phi$ is \emphi{positive} if it contains no negated variables.  We
say that $\phi$ is \emphi{planar} if the graph where we create a
vertex for each clause and variable, and an edge between each variable
node and clause node if the variable appears in this clause (either
positively or negatively), is a planar graph.  Moreover, the formula
$\phi$ is \emphi{rectilinear} planar if it has a planar embedding
where the variable vertices all lie on a horizontal line with
non-crossing three legged clauses above and below.  See figure on the
right.  Specifically, each vertex on the horizontal line will be
replaced by a unit disk, and an edge from a variable to a clause
consists of at most one vertical and at most one horizontal line
segment.  Such an embedding will be called a rectilinear planar
embedding.

\begin{problem}(\PRPOTSAT) An instance of Positive Rectilinear Planar
    One in Three 3SAT consists of a positive rectilinear planar $3CNF$
    formula $\phi$, along with a rectilinear planar embedding.  The
    problem is to determine if there is an assignment such that each
    clause has exactly one true variable.  We call such an assignment
    a satisfying assignment.
\end{problem}

Mulzer and Rote \cite{mr-mwtnph-08} prove that \PRPOTSAT is
\NP-Complete.  In the rest of this section we describe a reduction
from \PRPOTSAT to \LowerBoundedCenter.

Consider the rectilinear planar embedding of a \PRPOTSAT instance.  We
wish to think of this embedding as a circuit.  In this circuit each
variable $x$ acts as a signal generator which sends either a \True or
a \False signal to all the clauses which use $x$.  Each edge is then a
connecting wire, and clauses will act as logic gates which are
satisfied if and only if they receive one \True input and two \False
inputs.

The reduction will be to an instance of \LowerBoundedCenter with
$\lb=4$.  We will allow integer weighted points in the reduction.
Given an instance of \PRPOTSAT we will construct a point set such that
there will be a solution to the \lbc instance using unit disks if and
only if the \PRPOTSAT instance was satisfiable.

\subsection{Gadgets Used in Reduction}

\paragraph*{Idea.}
Consider placing points at every integer position along the $x$-axis,
each with weight $2$.  Now consider an instance of \LowerBoundedCenter
on this point set with a lower bound of $4$.  The points can be
covered using unit diameter disks by placing them over each successive
pair of points.  However, there are two ways to do this.  Either we
can place disks such that the left end of each disk covers a point at
an even integer position or such that it covers a point at an odd
integer position, see \figref{starter}.  Thus these two possible
solutions can be seen as either setting a variable to \True or \False.

In the following, a \emphi{wire} is a sequence of weight $2$ points
connecting two gadgets, with unit spacing in between adjacent points
such that no three points are contained within a disk of diameter
$<\sqrt{13}/2$.

\parpic[r]{%
   \begin{minipage}{0.4\linewidth}%
       \vspace{1.2cm}
       \includegraphics[width=.95\linewidth]{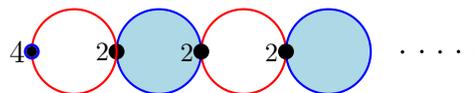}%
       \captionof{figure}{A variable gadget}%
       \figlab{starter}
   \end{minipage}}
\paragraph*{Signal Starter.}
For each variable we create a signal starter gadget.  This consists of
a horizontal chain of unit spaced points, where the weight of the
leftmost point is $4$ and each successive point has weight $2$.  A
valid solution to \lbc using unit disks can either cover the first
point of weight $4$ by itself or it can cover it and the first point
of weight $2$ together.  (See the blue and red solutions in the figure
to the right.)

\parpic[r]{%
   \begin{minipage}{0.38\linewidth}%
       \includegraphics[width=.97\linewidth]{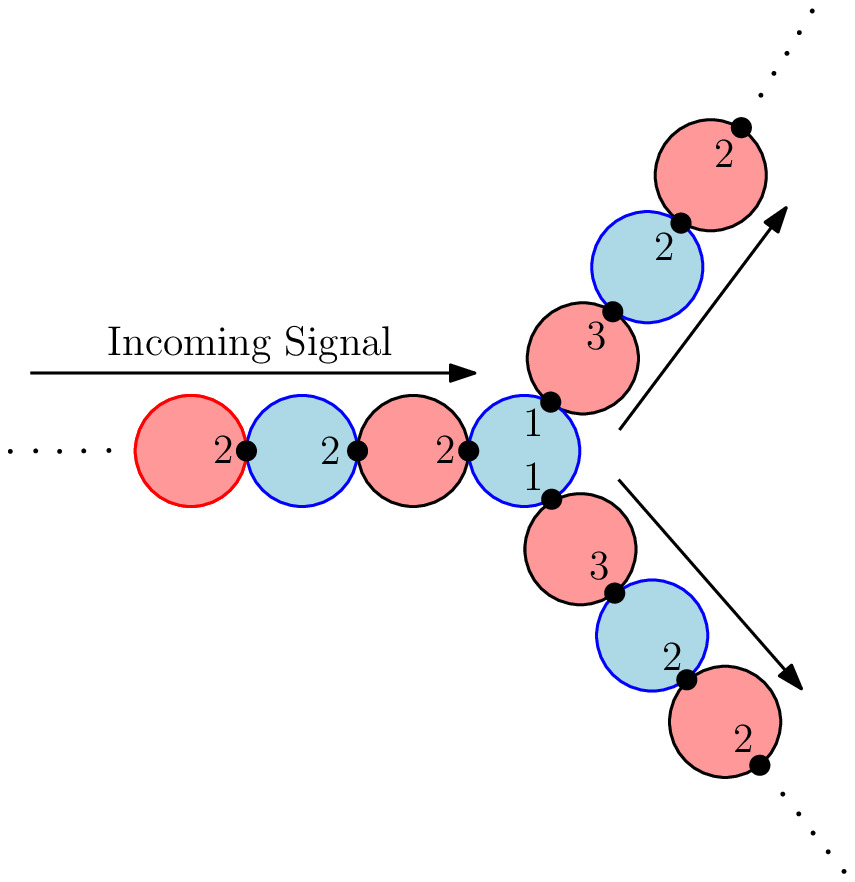}%
       \vspace{-1cm}%
       \captionof{figure}{} \figlab{split}
   \end{minipage}}

\paragraph*{Splitter.}
A variable might appear in multiple clauses and thus its signal will
need to be split multiple times so that each clause can get a copy.
\figref{split} shows how this is accomplished.  In the figure the
numbers next to each point represent its weight.  Observe that if one
is restricted to using unit diameter disks, then there are only two
possible feasible solutions, shown in red and blue in the figure.
Namely, if the red solution is the incoming signal then we know that
after the split, the outgoing signal must continue using the red disks
shown (and similarly for blue).

\vspace{.75cm}
\parpic[r]{%
   \begin{minipage}{0.32\linewidth}%
       \includegraphics[width=.9\linewidth]{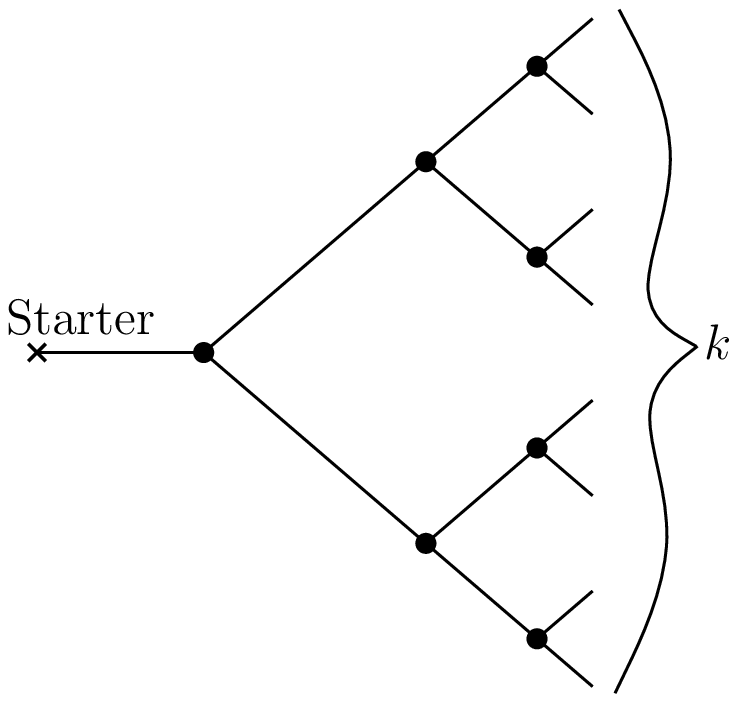}%
       \vspace{1cm}
   \end{minipage}}%
\vspace{-1cm}

\paragraph*{Variable Gadget.}
A variable gadget will consist of one signal starter gadget and $k-1$
signal splitter gadgets, where $k$ is the number of times the variable
appears in the given instance of \PRPOTSAT.  Specifically, we build a
binary tree out of the splitter gadgets where we attach the signal
starter gadget to the incoming wire of the splitter gadget
corresponding to the root of the tree.  The lengths of connecting
wires in the tree are such that there is a consistent signal exiting
from all the leaves.

\vspace{-0.13cm}%
 \parpic[r]{%
    \begin{minipage}{0.3\linewidth}%
        \includegraphics[width=.95\linewidth]{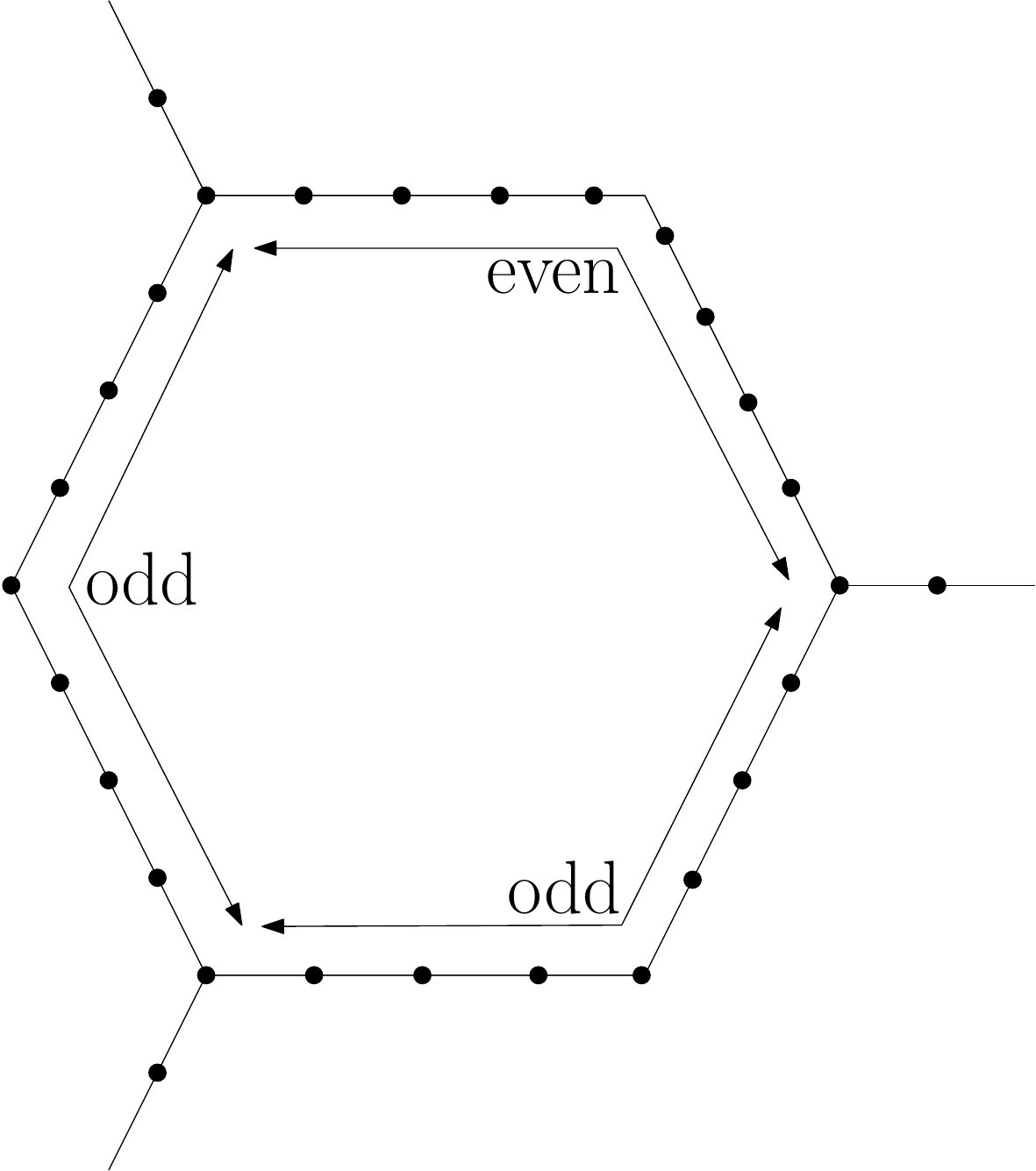}%
       \figlab{variable:2}
   \end{minipage}}

\paragraph*{Clause Gadget.}
Each clause has three variables and so it will have three incoming
wires from variable gadget leaves.  The parity of the length of these
wires will be the same for all variable clause pairs.  The clause
gadget will consist of a chain of points in the shape of a regular
hexagon (more precisely, it is close to a regular hexagon).  The wires
will attach to this hexagon at every other corner.  We require the
following properties from the points on the hexagon.

\begin{compactenum}[\quad(A)]
    \item All the points have weight 2.
    \item There is a point at each corner where a wire connects.
    \item The spacing between points along the hexagon and along the
    wire connecting to the hexagon is near unit length\footnote{Note
    that ideally we would always use unit spacing, though this may not
       be possible due to integrality and parity issues.  As will be
       explained more formally later, near unit distance will be good
       enough.}.

    \item The number of points on the hexagon in between adjacent
    incoming wires must be odd for two of the pairs, and even for the
    other pair. This can be done by slightly changing the length of
    the edges of the hexagon by slightly moving the relevant hexagon
    vertices into the center of the hexagon.

    This even/odd number of points guarantees that the clause can be
    satisfied if and only if the signal arrives on a single wire, see
    \figref{case:analysis} in \apndref{clause}.
\end{compactenum}

\parpic[r]{%
   \begin{minipage}{0.53\linewidth}%
     \begin{center}
       \includegraphics[scale=0.8]{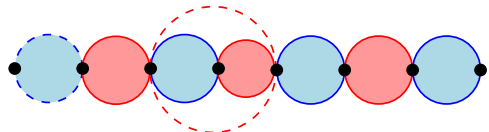}%
       \vspace{-0.3cm}
       \captionof{figure}{Dashed solution flips from blue to red.}
       \figlab{flip}
     \end{center}
   \end{minipage}
}%

\noindent\textbf{Distances along the wires.}
First, we make precise what we mean by a near-unit spacing. In order
to make the integrality and parity of the lengths of the wires and
clause gadgets work correctly we may not be able to always use
unit spacing.  However, all we require is that the spacing is such
that the signal being sent cannot flip.  Since there is a lower bound,
a flip can only happen if a disk takes too many of the points along
the wire or clause gadget.  So in order to prevent this for disks of
diameter $\leq \sqrt{13}/2$, all we require is that the spacing
between every other point along the wire or clause gadget be at least
$\sqrt{13}/2$. In particular, as demonstrated in the figure above,
along a wire we have flexibility, and we can allow the distances
between consecutive points to be slightly smaller than $1$.

\paragraph*{Putting the Gadgets Together.}
Let $G$ denote the rectilinear planar embedding of the \PRPOTSAT
instance.  Without loss of generality we may assume that the vertical
and horizontal segments which make up the edges of $G$ lie at integer
coordinates on a
grid.

\parpic[r]{%
   \begin{minipage}{0.378\linewidth}%
       \vspace{.8cm}
       \includegraphics[width=.97\linewidth]{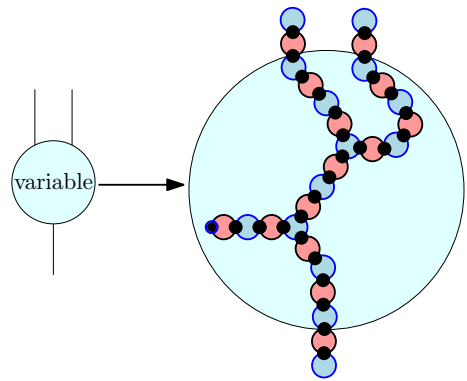}%
       \captionof{figure}{Variable gadget.}%
   \end{minipage}}

This embedding is transformed into an \lbc instance in the
natural way.  Namely, variables and clauses are replaced with variable
and clause gadgets, respectively, and edges are replaces with
connecting wires (see figure for an example of a variable
replacement).  Each time we insert a new gadget we may need to expand
the graph in order to fit in the new gadget.  Also, spacing between
edges may need to be increased so that points from different wires
have distance at least two, and so that wires can bend and connect to
gates properly without violating the properties that adjacent points
on a wire have near unit spacing and that no three points are within a
disk of diameter $\sqrt{13}/2$.  The reader should note that overall
the diameter of the embedding will only need to increase by a
polynomial factor. In particular, given an initial drawing of the
instance on a grid, one can initially scale it up by a polynomial
factor, and use the created space to construct all the gadgets
described above.

Overall, the resulting instance of \lbc can be computed in polynomial
time, and with a careful implementation the numbers used in the
representation are small (i.e., they each require $O( \log n)$ bits).

\subsection{Analysis}

\begin{lemma}
\lemlab{unit}
\RefProofInAppendix{unit}
    There is no polynomial time algorithm which computes
    the optimal solution to an instance of \LowerBoundedCenter unless
    $\P=\NP$.
\end{lemma}

\begin{proof:in:appendix:e}{\lemref{unit}}{unit}
    Consider an instance of \PRPOTSAT.  For the rectilinear planar
    embedding given we create a polynomial size instance of $\lbc$
    with $\lb=4$ as described in the previous section.  By
    construction we know that any solution to this instance of $\lbc$
    will require unit diameter disks.  We now prove that there is a
    solution using unit diameter disks if and only if the \PRPOTSAT
    instance had a satisfying solution.
    
    Call the point at which a wire connects to a clause gadget a
    \emphi{corner}.  We will say that a solution to the \lbc instance
    using unit diameter disks \emphi{covers} a corner if this point is
    covered along with the point just before it on the incoming wire.
    Otherwise we call the corner \emphi{exposed}.
    
    By the description of the gadgets from the previous section we
    know that, for a given variable, a solution using unit diameter
    disks either exposes all or covers all of its corresponding
    corners.  We now prove that the remaining points in a clause
    gadget can be covered using unit diameter disks if and only if
    exactly one corner was covered. See \figref{case:analysis} for an
    illustration of the following case analysis.
    
    Suppose the corner for a given wire is covered by that wire.  By
    simple case analysis one can show that if this is the case
    (regardless of which corner it was), then the remaining points on
    the hexagon of the clause gadget can be covered with unit diameter
    disks if and only if the other two wires leave their corners
    exposed.  Moreover, if all the wires leave their corner exposed,
    then there will be an odd number of points around the hexagon of
    the clause gadget and so there is no solution with unit diameter
    disks.  Therefore, by viewing a wire covering its corner as a
    \True and leaving it exposed as a \False, then the clause gadget
    can be covered with unit diameter disks if and only if there are
    two incoming \False's and one incoming \True.  Recall that this is
    exactly the requirement needed to satisfy each clause in the
    \PRPOTSAT instance. 
\end{proof:in:appendix:e}

\begin{lemma}
\lemlab{root}
\RefProofInAppendix{root}
    In the above construction, except for the unit disks specified in
    it, all other disks of weight at least $4$ must have radius at
    least $\geq \sqrt{13}/2$.
\end{lemma}

\begin{proof:in:appendix:e}{\lemref{root}}{root}
    Intuitively, such a disk of relatively small radius can happen
    only where three wires meet, which is always at $120$ degrees. As
    far as the construction is concerned, this corresponds to the wire
    flipping its signal.

    \parpic[r]{%
       \begin{minipage}{0.4\linewidth}%
           {\includegraphics[width=.83\linewidth]{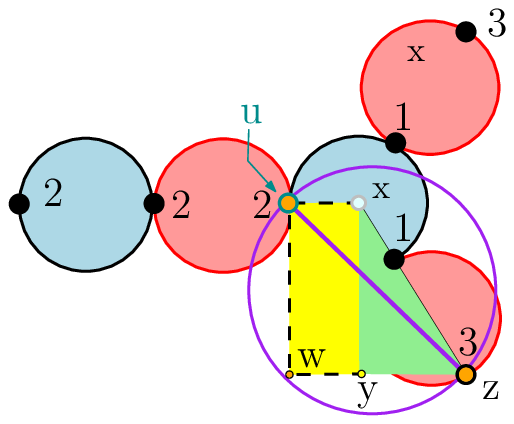}}%
           \vspace*{-0.4cm}
           \captionof{figure}{}%
           \figlab{flip2} \smallskip
       \end{minipage}}
    
    We now show that the signal can flip at a splitter only if disks
    of diameter $\geq \sqrt{13}/2$ are used.  Again since we have a
    lower bound, the problem is when a disk tries to take too many
    points.  This situation is shown in \figref{flip2}, where the
    optimal blue solution using unit disks instead attempts to take
    the larger disk covering the two orange points.  Observe that this
    is the smallest diameter disk which flips the signal and contains
    $\geq \lb$ points.  The diameter of this disk can be calculated by
    figuring out the dimensions of the pink triangle shown in the
    figure.  This is a right triangle with $\angle zxy = 30^\circ$.
    Therefore, since the length of the hypotenuse $xz$ is $3/2$, the
    length of $yz$ is $\distE{x z} \sin \angle zxy = \distE{xz}/2 =
    3/4$. As such, for the bottom of the triangle we have $\distE {x
       y} = \sqrt{\distE{xz}^2 - \distE{yz}^2} = \sqrt{9/4 - 9/16} =
    3\sqrt{3}/4$ and the height is $3/4$.  Therefore, since the yellow
    region has a height of $1/2$, we have that $\distE{uw} = \distE{xy
    } = 3\sqrt{3}/4$ and $\distE{wz} = 1/2 + 3/4 = 5/4$. As such,
    $\distE{uz} = \sqrt{ \distE{uw}^2 + \distE{wz}^2} =
    \sqrt{(3\sqrt{3}/4)^2+(5/4)^2} = \sqrt{27+25}/4 = \sqrt{13}/2$.
    Therefore, in order to guarantee that the signal does not flip,
    the solution must be restricted to disks of diameter
    $<\sqrt{13}/2$.

    One can verify, arguing similarly, that signal flipping at corners
    where wires connect to clause gadgets requires disks with diameter
    $\geq \sqrt{2+\sqrt{3}}>\sqrt{13}/2$. 
\end{proof:in:appendix:e}

\medskip

\begin{proofext}{\emph{\thmref{hardness}}}: %
    Consider an instance of \lbc obtained by converting a
    \PRPOTSAT instance as described in the previous section.  
    Clearly, our construction never admits a solution to the \lbc
    instance using smaller than unit diameter disks (regardless of
    whether the \PRPOTSAT instance was satisfiable).  We also know 
    by \lemref{unit} and \lemref{root}
    that there is a solution of diameter $d$, for any $1\leq d<
    \sqrt{13}/2$, if and only if the instance of \PRPOTSAT was
    satisfiable.  Therefore, if we could approximate the \lbc problem
    within a factor of $\sqrt{13}/2$ in polynomial time, then we could
    decide \PRPOTSAT in polynomial time.
\end{proofext}


\section{Conclusions}
We showed that \LowerBoundedCenter is both \APXHard and has a fast
constant factor approximation in low dimensional Euclidean space.  A
natural direction for future research is the lower bounded
version of the $k$-median problem.  A constant factor approximation
for this problem is already known \cite{as-iaglbfl-11,s-lbfl-10}.
Unlike \lbc, this problem is not known to be \APXHard,
though so far our current attempts to get a \PTAS have not panned out.

One can also consider adding an upper bound in addition to the lower bound
constraint, in essence specifying approximately the desired size of
the clusters.  It is not hard to verify that if the upper bound is at
least twice the lower bound then the algorithm presented in this paper
for \lbc carries over.  Otherwise, the problem seems to be
considerably more difficult.

\InConfVer{%

}
\InFullVer{%
   \bibliographystyle{alpha}%
}
\bibliography{lbc}%

\newpage
\appendix
\section{Clause Gadget Explained}
\apndlab{clause}
\vspace{-.5cm}
\begin{figure}[h!]
    \begin{tabular}{cc}
        \begin{minipage}{0.3\linewidth}
            \includegraphics[page=1,scale=.95]{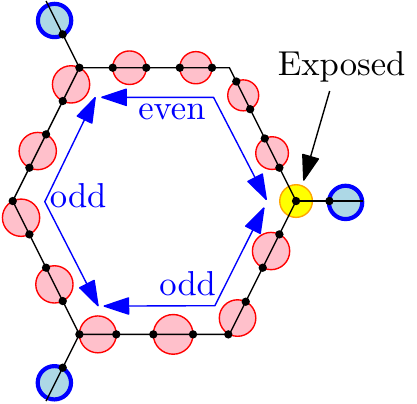}
        \end{minipage}
        &
        \hspace{.4cm}
        \begin{minipage}{0.6\linewidth}
        \vspace{.5cm}
            The signal does not arrive on any of the wires, and in any
            cover of this clause, one of the points is exposed (that
            is, the lower bound of $4$ can be met for this point, only
            by ``significantly'' enlarging one of the adjacent disks
            to also cover it).%
            \\[1cm]
        \end{minipage}
    \end{tabular}
    \vspace{-1.6cm}

    \begin{tabular}{cc}
        \begin{minipage}{0.6\linewidth}
            ~\\[1cm]
            If the signal arrives on only one of the wires, then one
            can cover the clause with disks of radius $1$ (and meet
            the lower bound requirement). The other cases of a signal
            arriving on a single wire follow by symmetry.
        \end{minipage}
        &
        \begin{minipage}{0.3\linewidth}
            \vspace{.5cm}
            \includegraphics[page=2,scale=.95]{figs/clause_assignment}
        \end{minipage}
    \end{tabular}

    \begin{tabular}{ccc}
        \\
        \includegraphics[page=3,scale=.95]{figs/clause_assignment}
        & 
        \includegraphics[page=4,scale=.95]{figs/clause_assignment}
        &
        \includegraphics[page=5,scale=.95]{figs/clause_assignment}
        \\[-0.5cm]
        (A) & (B) & (C)
    \end{tabular}
    \begin{center}
        \begin{minipage}{0.99\linewidth}
            If the signal arrives on two of the three wires, then
            again, one of the points must be exposed. It is easy to
            verify that in any of these cases, there is a portion of
            the clause between two gates that needs covering, but it
            contains an odd number of points -- thus the points can
            not be covered with disks of radius one.
        \end{minipage}
    \end{center}
    \begin{tabular}{cc}
        \begin{minipage}{0.3\linewidth}
            \includegraphics[page=6,scale=.95]{figs/clause_assignment}
        \end{minipage}
        &
        \begin{minipage}{0.6\linewidth}
            If the signal arrives on all three wires then the points
            of the clause can not be covered, by the same
            argumentation of case (B) above (for the reader's
            enjoyment, we show a different covering pattern in the
            figure).
        \end{minipage}
    \end{tabular}
     \caption{The clause gadget can be satisfied \si{iff} the signal
        arrives on only one wire. }
    \figlab{case:analysis}
\end{figure}


  \immediate\closeout\myoutfile 
  \section{Proofs}
  \input{fragment/myfile.tmp}


\end{document}